\documentclass[11pt,USletter]{article}
\usepackage[margin=1in]{geometry}
\usepackage{amssymb,amsmath,amsthm}
\usepackage{microtype}
\usepackage{verbatim,hyperref}
\usepackage{algorithm}
\usepackage{algpseudocode}
\usepackage{mathtools}
\usepackage{tikz}
\usepackage{eqlist}
\usepackage[justification=centering]{caption}

\newtheorem{theorem}{Theorem}
\newtheorem{lemma}{Lemma}
\newtheorem{observation}{Observation}
\newtheorem{corollary}{Corollary}

\theoremstyle{definition}

\theoremstyle{remark}

\def\polylog{\operatorname{polylog}}
\DeclareMathOperator{\opt}{OPT}
\DeclareMathOperator{\alg}{ALG}
\DeclareMathOperator{\algb}{ALG-B}
\DeclareMathOperator{\algbon}{ALG-B_{ON}}
\DeclareMathOperator{\algboff}{ALG-B_{OFF}}
\DeclareMathOperator{\algon}{ALG_{ON}}
\DeclareMathOperator{\algoff}{ALG_{OFF}}
\DeclareMathOperator{\lowerbound}{\Omega\left(\frac{1}{\epsilon}m^{\log\left(\frac{3 + 2\epsilon}{2 + 2\epsilon}\right)}\right)}
\newcommand{\doubleR}{\mathbb{R}}
\newcommand{\doubleN}{\mathbb{N}}

\newcommand{\offlinemetric}{\textit{Minimum Metric Perfect Matching}}
\newcommand{\str}[1] {#1^*}

\algnewcommand\algorithmicforeach{\textbf{for each}}
\algdef{S}[FOR]{ForEach}[1]{\algorithmicforeach\ #1\ \algorithmicdo}

\title{Deterministic Min-Cost Matching with Delays}
\author{
    Yossi Azar\thanks{School of Computer Science, Tel Aviv University, Israel; azar@tau.ac.il}\and
    Amit Jacob-Fanani\thanks{School of Computer Science, Tel Aviv University, Israel; amitj@mail.tau.ac.il}
}
\date{}

\begin{document}

\maketitle

\begin{abstract}
%abstract
We consider the online
Minimum-Cost Perfect Matching with Delays (MPMD) problem introduced by Emek et al. (STOC 2016),
in which a general metric space is given, 
and requests are submitted in different times in this space by an adversary.
The goal is to match requests, while minimizing the sum of distances between matched pairs
in addition to the time intervals passed from the moment each request appeared until it is matched.

In the online Minimum-Cost Bipartite Perfect Matching with Delays (MBPMD) problem
introduced by Ashlagi et al. (APPROX/RANDOM 2017), 
each request is also associated with one of two classes,
and requests can only be matched with requests of the other class.

Previous algorithms for the problems mentioned above, 
include randomized $O\left(\log n\right)$-competitive algorithms
for known and finite metric spaces, $n$ being the size of the metric space,
and a deterministic $O\left(m\right)$-competitive algorithm,
$m$ being the number of requests.

We introduce $O\left(m^{\log\left(\frac{3}{2}+\epsilon\right)}\right)$-competitive 
deterministic algorithms for both problems and for any fixed $\epsilon > 0$. In particular, for a small enough $\epsilon$
the competitive ratio becomes $O\left(m^{0.59}\right)$.
These are the first deterministic algorithms for the mentioned online matching problems,
achieving a sub-linear competitive ratio. 
Our algorithms do not need to know the metric space in advance.

\end{abstract}

\section{Introduction}
In the algorithmic graph theory, a \textit{Perfect Matching} is a subset of graph edges, 
in which each vertex of the graph is incident on exactly one edge of the subset,
and the weight of the matching is the sum of the weights of the edges of the subset.
In the well known \textit{Minimum-Cost Perfect Matching} problem a weighted graph is given,
and a \textit{Perfect Matching} of minimum weight is to be found.
The \textit{Blossom Algorithm} due to Edmonds~\cite{Edmonds_CJM65} 
is the first algorithm to solve this problem in polynomial time.

Many versions of the \textit{Minimum-Cost Perfect Matching} 
problem have been studied over the last few decades,
some of the noticeable variants are online versions of the problem 
(e.g. \textit{Minimum-Cost Perfect Matchings with Online Vertex Arrival} 
due to Kalyanasundaram and Pruhs~\cite{KalyanasundaramP_JAlgorithms93}).

In this paper we suggest a deterministic algorithm for the 
\textit{Minimum-Cost Perfect Matching with Delays} (MPMD) variant, 
which was introduced by Emek et al.~\cite{EmekKW_STOC16},
and a similar deterministic algorithm for another variation of the problem -
the \textit{Minimum-Cost Bipartite Perfect Matching with Delays} (MBPMD) problem,
which was introduced by Ashlagi et al.~\cite{Ashlagi17}.

To illustrate the MPMD problem, 
imagine players logging in through a server to an online game at different times,
unknown a priori to the server they have connected through.
The server then needs to match between the players while maximizing
their satisfaction from playing the game.
Players feel satisfied when they play against players at a level similar to their own.
Therefore, when pairing players, the server needs to consider the difference in levels between the players, called the \textit{connection cost}.

Once logged in, a player doesn't necessarily start playing instantly,
as the server can postpone the decision regarding with whom to match the player, until a good match is found (i.e. another player at a similar level logs in to the game).
This is a poor strategy since players are unhappy when forced to wait too long until they start playing.
The time a player has to wait until the game starts is called the \textit{delay cost}.

More formally, an adversary presents requests at points in a general metric space,
in an online manner.
The goal is to produce a minimum-cost perfect matching
when the cost of an edge is the sum of its \textit{connection cost} (the distance between the two points in the metric space) 
and the \textit{delay cost} of the two requests matched by the edge.
All requests have to be matched by the server after a finite time from the moment they have arrived.

The MBPMD problem is an extension of the MPMD problem (due to Ashlagi et al.~\cite{Ashlagi17}),
in which each of the requests may take one of two colors,
and each edge of the matching, must be incident on one request from each color.
The MBPMD problem has many applications, such as matching drivers to passengers (Uber, Lyft), job finding platforms, etc. 

\paragraph*{Background.}
The standard method used to measure an online algorithm's performance is its competitive ratio.
We use this method when comparing the performance of matching algorithms for both MPMD and MBMPD.
An algorithm is $\alpha$-competitive 
if the maximum ratio between the cost of the algorithm
to the cost of the optimum solution, over all inputs, is bounded by $\alpha$.

The first algorithm for MPMD was developed by Emek et al.~\cite{EmekKW_STOC16} with an expected competitive ratio $O\left(\log^2n+\log\Delta\right)$
on a finite metric space of size $n$, where $\Delta$ is the aspect-ratio of the metric space
(the ratio of the maximum distance to the minimum distance between any two points in the metric space).
Azar et al.~\cite{Azar17} improved the competitive ratio to $O\left(\log n\right)$,
and showed a lower bound of $\Omega\left(\sqrt{\log n}\right)$ (both deterministic and randomized).
Ashlagi et al.~\cite{Ashlagi17} improved this lower bound to $\Omega\left(\frac{\log n}{\log\log n}\right)$ (both deterministic and randomized).
They also gave an $O\left(\log n\right)$-competitive randomized algorithm for MBPMD.

All mentioned above algorithms are randomized (on a general finite metric).
In online algorithms where one cannot repeat the algorithm in case the cost is high, 
a deterministic algorithm is preferable.
Bienkowski et al.~\cite{BienkowskiKS17} provided the first deterministic algorithm for MPMD on general metrics,
with a competitive-ratio of $O\left(m^{2.46}\right)$, $m$ being the number of requests.
While the previous algorithms require the metric space to be known a priori, 
their algorithm does not, and is also applicable when the metric space is revealed in an online manner.
Bienkowski et al. also noted that the algorithm of~\cite{Azar17} 
can be used to provide an $O\left(n\right)$-competitive deterministic algorithm
for a general known metric space.
Recently, Bienkowski et al.~\cite{Bienkowski18} provided a new primal-dual deterministic algorithm
for MPMD on general metrics,
with a competitive-ratio of $O\left(m\right)$, $m$ being the number of requests.

Prior to our result there was no deterministic sub-linear competitive algorithm, 
neither in $n$ nor in $m$.

\paragraph*{Our Contribution.}
In this paper we introduce deterministic algorithms for both versions of the problem,
both with a competitive ratio $O\left(\frac{1}{\epsilon}m^{\log\left(\frac{3}{2}+\epsilon\right)}\right)$. 
When the constant $\epsilon$ is small enough,
this becomes $O\left(m^{0.59}\right)$.
Our algorithms do not need to know the metric space in advance.

We present a simple algorithm, 
which is an adaptation of the greedy algorithm for the \textit{Minimum-Cost Perfect Matching} problem by Reingold and Tarjan~\cite{ReingoldT81} to an online environment.
In our algorithm, requests grow hemispheres around them in a metric 
that is the Cartesian product of the original metric and the time axis (also called the \textit{time-augmented metric space}).
The hemispheres radii grow slowly in the negative direction of the time axis.
Once a request is found on the boundary of another request's hemisphere,
they are matched by the algorithm. Our analysis is inspired by the analysis of 
the original greedy algorithm by Reingold and Tarjan.

In the bipartite case, the algorithm is essentially the same,
but requests are matched only if they are of different colors.

\paragraph*{Related Work.}
First we consider related work \textbf{with delays}.
Since Emek et al.~\cite{EmekKW_STOC16} introduced the notion of online problems with delayed service,
there has been a growing number of works studying such problems~ (e.g. 
\textit{Online Service with Delays} \cite{Azar17Service},
\textit{Minimum-Cost Bipartite Perfect Matching with Delays}~\cite{Ashlagi17},
\textit{Minimum-Cost Perfect Matching with Delays for Two Sources}~\cite{EmekSW17}).
Works dealing with the \textit{Minimum-Cost Perfect Matching with Delays} and \textit{Minimum-Cost Bipartite Perfect Matching with Delays} problems,
such as the papers by 
Emek et al.~\cite{EmekKW_STOC16}, 
Azar et al.~\cite{Azar17}, 
Ashlagi et al.~\cite{Ashlagi17} 
and Bienkowski et al.~\cite{BienkowskiKS17},
are the most closely related to this work.
As mentioned above, Emek et al.~\cite{EmekKW_STOC16} provided a randomized $O\left(\log^2n+\log\Delta\right)$-competitive algorithm for MPMD on general metrics,
in which $n$ is the size of the metric space and $\Delta$ is the aspect ratio.
They consider the randomized embeddings of the general metric space
into a distribution over metrics given by hierarchically separated full binary trees,
with distortion $O\left(\log n\right)$, and give a randomized algorithm for the
hierarchically separated trees metrics.

Subsequently, Azar et al.~\cite{Azar17} provided a randomized $O\left(\log n\right)$-competitive algorithm for the same problem, 
thus improving the original upper bound. 
They used randomized embedding of the general metric space
into a distribution over metrics given by hierarchically separated trees 
of height $O\left(\log n\right)$, with distortion $O\left(\log n\right)$.
Then they give a deterministic $O\left(1\right)$-space-competitive (that is the competitive ratio associated with the \textit{connection cost})
and $O\left(h\right)$-time-competitive (that is the competitive ratio associated with the \textit{delay cost}) 
algorithm over tree metrics, where $h$ is the height of the tree. 
This yields a competitive ratio of $O\left(\log n\right)$.
Moreover, they provided a randomized $\Omega\left(\sqrt{\log n}\right)$ lower bound, 
confirming a conjecture made by Emek et al.~\cite{EmekKW_STOC16} that the competitive ratio of any online
algorithm for the problem must depend on $n$.

Ashlagi et al.~\cite{Ashlagi17} improved the lower bound
on the competitive ratio to $\Omega\left(\frac{\log n}{\log\log n}\right)$, 
almost matching the upper bound of Azar et al. of $O\left(\log n\right)$.
The rest of the paper focuses on the bipartite version of the problem,
providing an $O\left(\log n\right)$-competitive ratio by the adaptation of the algorithm
of Azar et al.~\cite{Azar17} to the bipartite case.

In order to provide a \textit{deterministic} algorithm, 
Bienkowski et al.~\cite{BienkowskiKS17} used a different approach for the problem - 
they used a semi-greedy scheme of a ball-growing algorithm. 
%that matches two requests when the spheres surrounding them meet each other,
%while also forcing the waiting times for paired requests
%to be no more than a constant factor apart.
In their analysis, they fix an optimal matching, 
and charge the cost of each matching-edge generated by their algorithm
against the cost of an existing matching-edge of the optimal matching.
As mentioned above, their algorithm achieves a competitive ratio of $O\left(m^{2.46}\right)$, 
where $m$ is the number of requests.

Bienkowski et al. improved this result in~\cite{Bienkowski18} 
by providing a new $O(m)$-competitive LP-based algorithm. 
Briefly, their algorithm maintains a primal relaxation of the matching problem and its dual 
(the programs evolve in time as more requests arrive). 
Dual variables are increased along time, until a dual constraint (corresponding to a pair of requests) 
becomes tight, which results in the algorithm connecting the pair. 
They also proved that their analysis is tight (the competitive-ratio of their algorithm is $\Omega(m)$).
Recall that our algorithm acheives a sub-linear competitive-ratio (in $m$).

Next we consider related work \textbf{without delays}.
The \textit{Online Minimum Weighted Bipartite Matching} (OMM) problem 
due to~\cite{KalyanasundaramP_JAlgorithms93, KhullerMV91} is another important online 
version of the \textit{Minimum-Cost Perfect Matching} problem,
in which $k$ vertices are given a priori, 
and $k$ additional vertices are revealed at different times, 
together with the distances from the first $k$ vertices.
The algorithm then needs to match the later $k$ vertices
to the first $k$ vertices, while trying to minimize
the total weight of the produced matching.
In this version, delay of the algorithm's decision is not available.
Kalyanasundaram and Pruhs~\cite{KalyanasundaramP_JAlgorithms93}
and Khuller et al.~\cite{KhullerMV91} showed independently
a tight upper and lower bounds of $2k-1$ 
on the deterministic competitive ratio of the problem.

The first sub-linear competitive randomized algorithm for the problem,
was given by Meyerson et al.~\cite{MeyersonNP06} using randomized embeddings into trees, 
with a competitive ratio of $O(\log^3 k)$. 
Consequently, Bansal et al.~\cite{BansalBGN_Algorithmica14} improved this upper bound
by providing a $O(\log^2k)$-competitive randomized algorithm. 
In addition, they showed an $\Omega(\log k)$ lower bound on the competitive ratio
for randomized algorithms.

The special case of line-metrics 
is argued to be the most interesting instance of OMM (e.g.~\cite{KoutsoupiasN03}).
Kalyanasundaram and Pruhs conjectured in 1998~\cite{kalyanasundaram_pruhs_1998} 
that there exists a 9-competitive deterministic
algorithm for OMM on line-metrics,
but in 2003 Fuchs et al.~\cite{FuchsHK03} disproved the conjecture,
proving a lower bound of 9.001 for deterministic algorithms.
This is the best known lower bound thus far.

Antoniadis et al.~\cite{AntoniadisBNPS14} presented the first
sub-linear deterministic algorithm for line-metrics,
with a competitive ratio of $O\left(\frac{1}{\epsilon}k^{\log\left(\frac{3}{2}+\epsilon\right)}\right)$.
Recently, Nayyar and Raghvendra \cite{NayyarR17}
improved this upper bound to $O(\log^2 k)$ by careful analysis of 
the deterministic algorithm present in \cite{raghvendra16}.
Gupta and Lewi~\cite{GuptaL12} provided a randomized $O(\log k)$-competitive algorithm
for doubling metrics, hence for line-metrics as well.

To summarize, 
the best known deterministic upper bound on the competitive ratio for line-metrics is $O(\log^2 k)$,
and best known lower bound is 9.001.
For randomized algorithms the best known upper bound is $O(\log k)$.

\paragraph*{Paper Organization.}
We describe the algorithm for \textit{Minimum-Cost Perfect Matching with Delays} in Section~\ref{sec_det_alg}
and analyze its performance in Section~\ref{sec_analysis}.
Through an example in Appendix~\ref{sec_lower_bound} we show that our analysis is tight, 
and prove that the competitive ratio of our algorithm indeed depends on the number of requests, 
and not on the size of the metric space. 
In addition, we show in Appendix~\ref{sec_time} that minor natural changes to the algorithm, 
do not transform the competitive ratio into a function of the size of the metric space 
(in the case of a finite metric space) instead of the number of requests.
In Section~\ref{sec_bipartite} we present the algorithm for \textit{Minimum-Cost Bipartite Perfect Matching with Delays} and analyze its performance.

\section{Preliminaries}\label{sec_prelim}
A \textit{metric space} $\mathcal{M} = (S,d)$
is a set $S$ and a distance function $d:S\times S\longrightarrow\mathbb{R}^+$
that meets the following conditions: non-negativity, symmetry, the triangle-inequality, and that $d(x, y) = 0$ if and only if $x = y$.
% \begin{itemize}
%     \item{\makebox[2.5cm]{$\forall x \in S$\hfill} $d(x,x) = 0$}
%     \item{\makebox[2.5cm]{$\forall x\neq y \in S$\hfill} $d(x,y)>0$}
%     \item{\makebox[2.5cm]{$\forall x,y \in S$\hfill} $d(x,y)=d(y,x)$}
%     \item{\makebox[2.5cm]{$\forall x,y,z \in S$\hfill} $d(x,y) + d(y,z)\geq d(x,z)$}
% \end{itemize} 
When $S$ is finite, we refer to $\mathcal{M}$ as a \textit{finite metric space},
and an \textit{infinite metric space} otherwise.
\subsection{Model}
In the online \textit{Minimum-Cost Perfect Matching with Delays} problem on a metric space
$\mathcal{M}=(S,d)$ (known a priori to the algorithm),
an input instance $\mathcal{I}=\langle r_i\rangle_{i=1}^m$ 
is presented to the algorithm in an online fashion,
so that each request $r_i$ is revealed to the algorithm at time $t(r_i)$ at the location $x(r_i)\in S$.
The number of requests $m$ is even and unknown a priori to the algorithm.

The online algorithm should produce a perfect matching in real time.
Formally, two requests $p, q$ can be matched by the algorithm at any time $t \geq \max(t(p), t(q))$, 
if they have not been matched yet by the algorithm.

Let $\langle p_i,q_i,t_i\rangle_{i=1}^\frac{m}{2}$ be the set of pairs of requests matched by the algorithm,
and their matching times ($p_i$ and $q_i$ were matched by the algorithm at $t_i$),
then the cost of the matching produced by the algorithm is
\begin{equation*}
    \sum_{i=1}^{\frac{m}{2}} d(x(p_i), x(q_i)) + |t_i - t(p_i)| + |t_i - t(q_i)| 
\end{equation*}
In other words, 
the cost is the sum of the connection cost of all matched pairs 
in addition to the sum of the delay cost of all requests.
The goal of the algorithm is to minimize this cost.

The \textit{Minimum-Cost Bipartite Perfect Matching with Delays}
is virtually the same problem as the \textit{Minimum-Cost Perfect Matching with Delays} problem, 
except that each request $r_i$ is associated with one of two classes,
so that each request $r_i$ can be matched to a request $r_j$ 
if and only if $class(r_i)\neq class(r_j)$.

\subsection{The time-augmented metric space}
Given a metric space $\mathcal{M}=(S,d)$ 
define the \textit{time-augmented metric space} as $\mathcal{M}_T=(S\times \doubleR, D)$ 
where $D$ is a distance function defined as
\begin{equation*}
    D\left(\left(l_1,t_1\right),\left(l_2,t_2\right)\right) = d(l_1,l_2) + |t_1 - t_2|
\end{equation*} 
assuming $(l_1,t_1), (l_2,t_2)\in S\times \doubleR$.
That is, the time axis was added as another dimension in the metric space.
One can easily verify that $D$ indeed defines a metric.

The following Lemma shows that for offline algorithms, 
solving the Minimum-Cost Perfect Matching \textbf{with Delays} 
problem in the metric space $\mathcal{M}$ is equivalent to solving the 
Minimum-Cost Perfect Matching problem in $\mathcal{M}_T$.
\begin{lemma}\label{lem_opts}
    Assume $\mathcal{I}=\langle r_i\rangle_{i=1}^m$ is an instance of MPMD then
    $\opt$ can be computed as the weight 
    of an optimal solution for the \offlinemetric{} problem on the instance $\mathcal{I}$ 
    as points in the time-augmented metric space $\mathcal{M}_T$.
\end{lemma}
\begin{proof}
    Let $\str{\opt}$ be an optimal solution for \offlinemetric{} over the instance $\mathcal{I}$. We show that $\opt = \str{\opt}$.

    Let $A$ be the solution for \offlinemetric{} over the instance $\mathcal{I}$,
    which matches the pairs corresponding to those matched by $\opt$.
    The cost of $A$ is at most the cost of $\opt$, 
    since for a given pair $(u,v)$ matched by $\opt$ at time $t_{uv} \geq \max\left(t(u),t(v)\right)$,
    $\opt$ would pay $t_{uv} - t(u) + t_{uv} - t(v) + d(x(u),x(v))$, 
    while $A$ would pay $D(u,v) = |t(u)-t(v)| + d(x(u),x(v))$ which cannot be larger.
    Therefore $\str{\opt} \leq A \leq \opt$.
    
    For the other direction we define an online algorithm $B$
    which matches the pairs corresponding to those matched by $\str{\opt}$, 
    as soon as the two end-points arrive.
    For a given pair of requests $(p,q)$ matched by $B$, 
    it pays
    \begin{equation*}
        \max(t(p),t(q)) - t(p) + \max(t(p),t(q)) - t(q) + d(x(p),x(q)) = |t(p)-t(q)| + d(x(p),x(q))
    \end{equation*}
    Therefore the cost paid by $B$ is the same as the cost paid by $\str{\opt}$.
    
    Hence $\opt \leq B = \str{\opt}$.
\end{proof}

\section{A Deterministic Algorithm for MPMD on General Metrics}\label{sec_det_alg}
Our algorithm (ALG($\epsilon$)) is parametrized with a constant $\epsilon \in \doubleR$.
Upon the arrival of a request $p \in S \times \doubleR$, the algorithm begins to grow 
a hemisphere surrounding $p$ in the negative direction of the time axis, 
such that the radius growth rate is $\epsilon$.
Therefore, at time $t$, 
a request $q \in S \times \doubleR$ is on the hemisphere's boundary 
if and only if $\epsilon\left(t-t\left(p\right)\right)=D(p,q)$ and $t(q)\leq t(p)$, 
where $D$ is the distance function defined by the time-augmented metric space $\mathcal{M}_T$.
The algorithm matches a request $q$ to a request $p$ as soon as $q$ is found on the boundary of $p$'s hemisphere.

Note that the algorithm does not need to know the metric space in advance,
but it only requires that together with any arriving request $p$,
it learns the distances from $p$ to all previous requests.

\begin{algorithm}
    \caption{A Deterministic Algorithm for MPMD on General Metrics\label{alg:mono}}
    \begin{algorithmic}[1]
        \Procedure{ALG}{$\epsilon$}
        \State \textbf{At every moment $t$:}
        \State Add the new requests that arrive at time $t$
        \ForEach{unmatched request $p$}
            \ForEach{unmatched request $q \neq p$}
                \If{$t(p) \geq t(q)$ \textbf{and} $t = t(p) + \frac{D\left(p,q\right)}{\epsilon}$}
                \State \textbf{match}$\left(p,q\right)$
                \EndIf
            \EndFor
        \EndFor
     \EndProcedure
    \end{algorithmic}
\end{algorithm}

The algorithm is described as a continuous process 
but can be easily discretized using priority queues
over anticipated matching events for each pair.

The algorithm breaks ties arbitrarily (i.e. a request that is on multiple hemispheres at the same time, or multiple requests that are on the same hemisphere).
Note that for the analysis of the algorithm we may assume that there are no ties, 
as an adversary might slightly perturb the points so that the algorithm would choose the worse option.

\subsection{Analysis}\label{sec_analysis}
\begin{theorem}\label{thm_main}
    $\alg(\epsilon)$ is $O\left(\frac{1}{\epsilon}m^{\log\left(\frac{3+\epsilon}{2}\right)}\right)$-competitive.
\end{theorem}
Given $\epsilon\in\doubleR$ we run $\alg(\epsilon)$ 
over the instance $\mathcal{I}=\langle r_i\rangle_{i=1}^m$, 
that is with a hemisphere growth rate of $\epsilon$.
For the analysis, we denote $\algon$ to be the cost paid by $\alg(\epsilon)$, 
and $\algoff$ to be the weight of the matching produced by $\alg(\epsilon)$, 
when viewing $\mathcal{I}$ as points in the time-augmented metric space $\mathcal{M}_T$.
$\opt$ is the cost of an optimal solution for MPMD over the instance $\mathcal{I}$.

Consider the last two pairs of requests to be matched by $\alg$.
They consist of four requests, name them $a, b, c, d$, 
such that $(a, b)$ is one pair, and $(c,d)$ is the second pair.
Assume w.l.o.g that $(a,b)$ were matched at time $t_{ab}$,
and $(c,d)$ at $t_{cd} \geq t_{ab}$.
Also, assume w.l.o.g that $t(a) \leq t(b)$.
\begin{lemma}\label{lem_main}~ %
    \begin{enumerate}
        \item $D(a,b) \leq (1+\epsilon)D(a,c)$ and $D(a,b) \leq (1+\epsilon)D(a,d)$\label{lem_main:1}
        \item $D(a,b) \leq (1+\epsilon)D(b,c)$ and $D(a,b) \leq (1+\epsilon)D(b,d)$\label{lem_main:2}
    \end{enumerate}
\end{lemma}
\begin{proof}
    We only prove $D(a,b) \leq (1+\epsilon)D(a,c)$ and $D(a,b) \leq (1+\epsilon)D(b,c)$ 
    since there is no difference between $c$ and $d$.

    To prove (\ref{lem_main:1}), we look at two cases, that are $t(c) \geq t(a)$, and $t(c) < t(a)$.
    
    \textbf{Case $t(c) \geq t(a)$:}
    Upon the arrival of $c$ and $b$,
    the algorithm begins to grow hemispheres surrounding them,
    and in particular $a$ might be on their boundaries.
    Since $(a,b)$ was the first pair to be matched,
    $a$ was on $b$'s hemisphere before it was on $c$'s hemisphere
    (otherwise $(a,c)$ should have been matched first).
    Therefore $t(b) + \frac{D(a,b)}{\epsilon} \leq t(c) + \frac{D(a,c)}{\epsilon}$, and we conclude
    \begin{equation*}
        D(a,b)\leq D(a,c) + \epsilon(t(c) - t(b)) \leq D(a,c) + \epsilon(t(c) - t(a)) \leq (1+\epsilon)D(a,c)
    \end{equation*}
    
    \textbf{Case $t(c) < t(a)$:}
    Upon the arrival of $a$ and $b$,
    the algorithm begins to grow hemispheres surrounding them.
    In particular, $a$ might be on the boundary of $b$'s hemisphere, 
    and $c$ might be on the boundary of $a$'s hemisphere.
    Since $(a,b)$ was the first pair to be matched, 
    $a$ was on $b$'s hemisphere before $c$ was on $a$'s hemisphere 
    (otherwise $(a,c)$ should have been matched first).
    Therefore $t(b) + \frac{D(a,b)}{\epsilon} \leq t(a) + \frac{D(a,c)}{\epsilon}$.
    Thus, we conclude that
    \begin{equation*}
        D(a,b)\leq D(a,c) + \epsilon(t(a) - t(b)) = D(a,c) - \epsilon(t(b) - t(a)) \leq D(a,c) \leq (1+\epsilon)D(a,c)
    \end{equation*}
    
    To prove (\ref{lem_main:2}), we look at the two cases $t(c) \geq t(b)$, and $t(c) < t(b)$.

    \textbf{Case $t(c) \geq t(b)$:}
    Upon the arrival of $c$ and $b$,
    the algorithm begins to grow hemispheres surrounding them.
    In particular, $a$ might be on the boundary of $b$'s hemisphere, 
    and $b$ might be on the boundary of $c$'s hemisphere.
    Since $(a,b)$ was the first pair to be matched,
    $a$ was on $b$'s hemisphere before $b$ was on $c$'s hemisphere 
    (otherwise $(b,c)$ should have been matched first).
    Therefore $t(b) + \frac{D(a,b)}{\epsilon} \leq t(c) + \frac{D(b,c)}{\epsilon}$.
    Thus, we conclude that
    \begin{equation*}
        D(a,b)\leq D(b,c) + \epsilon(t(c) - t(b)) \leq D(b,c) + \epsilon D(b,c) = (1+\epsilon)D(b,c)
    \end{equation*}
    
    \textbf{Case $t(c) < t(b)$:}
    Upon $b$'s arrival,
    the algorithm begins to grow a hemisphere surrounding it,
    and in particular $a$ and $c$ might be on its boundary.
    Since $(a,b)$ was the first pair to be matched,
    $a$ was on $b$'s hemisphere before $c$ was (otherwise $(b,c)$ should have been matched first).
    Therefore $t(b) + \frac{D(a,b)}{\epsilon} \leq t(b) + \frac{D(b,c)}{\epsilon}$.
    Thus, we conclude that
    \begin{equation*}
        D(a,b)\leq D(b,c) \leq (1+\epsilon)D(b,c)
    \end{equation*}
\end{proof}
We use the following well known observation.
\begin{observation}\label{obs1}
    The union of any two matchings is a set of vertex-disjoint cycles. 
    In every such cycle, the edges alternate between the two matchings.
    Note that two parallel edges are considered a cycle.
\end{observation}

Let $\mathcal{C}=\{C_1,\ldots,C_k\}$ be the set of cycles (vertices and edges) 
generated from taking the union of the matchings produced by $\alg$ and $\opt$.
Define $l_1,\ldots,l_k \in \doubleR$ such that $l_i$ is the total length of edges of $\alg$ in $C_i$.
Define similarly $\str{l_1},\ldots,\str{l_k} \in \doubleR$ for edges of $\opt$. 

\begin{lemma}\label{lem_cycles}
    $\frac{\algoff}{\opt} \leq \max_{i}\frac{l_i}{\str{l_i}}$
\end{lemma}
\begin{proof}
    \begin{equation*}
        \frac{\algoff}{\opt} = 
        \frac{\sum_{i=1}^k{l_i}}{\sum_{i=1}^k{\str{l_i}}} = 
        \sum_{j=1}^k{\frac{\str{l_j}}{\sum_{i=1}^k{\str{l_i}}}\frac{l_j}{\str{l_j}}} \leq 
        \sum_{j=1}^k{\frac{\str{l_j}}{\sum_{i=1}^k{\str{l_i}}}\max_r{\frac{l_r}{\str{l_r}}}} = 
        \max_r{\frac{l_r}{\str{l_r}}} \sum_{j=1}^k{\frac{\str{l_j}}{\sum_{i=1}^k{\str{l_i}}}} = 
        \max_r{\frac{l_r}{\str{l_r}}}
    \end{equation*}
\end{proof}
\begin{lemma}\label{lem_local_opt_and_greedy}
    Denote $\hat{\str{l_i}}$ the cost paid by an optimal algorithm for \offlinemetric{} 
    on the instance constructed from the vertices of $C_i$, 
    and $\hat{l_i}$ the cost of running $\alg$ over the vertices of $C_i$.
    Then $\hat{\str{l_i}} = \str{l_i}$ and $\hat{l_i} = l_i$.
\end{lemma}
\begin{proof}
    To prove $\hat{\str{l_i}} = \str{l_i}$ assume by contradiction that $\str{l_i} < \hat{\str{l_i}}$. 
    Notice that the subset of edges of $\opt$ contained in $C_i$ is a legal solution for \offlinemetric{} with cost $\str{l_i}$.
    Clearly $\str{l_i}$ is less than $\hat{\str{l_i}}$, contradicting the definition of $\hat{\str{l_i}}$.
    For the other direction, 
    let $E$ be the edges matched by $\opt$,
    and $\hat{E}$ be the edges matched by an optimal algorithm for \offlinemetric{} on the instance constructed from the vertices of $C_i$.
    Define $\bar{E} = (E \setminus C_i) \cup \hat{E}$.
    Notice that $\bar{E}$ is a legal solution for \offlinemetric{} on the instance $\mathcal{I}$
    with cost $\sum_{i=1}^k{\str{l_i}} - \str{l_i} + \hat{\str{l_i}} < \opt$ contradicting the definition of $\opt$.
    Therefore $\str{l_i} = \hat{\str{l_i}}$.

    To prove $\hat{l_i} = l_i$ we show that $K$ - 
    the matching produced by $\alg$ when running over the vertices of $C_i$, 
    is the same as $E_i$ - the subset of edges matched by $\alg$ and contained in $C_i$, 
    when running on the instance $\mathcal{I}$.
    Let $r = \frac{|C_i|}{2}$ where $|C_i|$ is the number of edges in $C_i$, 
    and note that $|E_i| = r = |K|$, since both $E_i$ and $K$ are matchings over $C_i$.
    Sort the edges of $E_i$ by the time they are formed from first to last: 
    $e_1 = (u_1,v_1),\ldots, e_r = (u_r, v_r)$, 
    and the same for the edges of $K$: $k_1=(p_1,q_1),\ldots,k_r=(p_r,q_r)$.
    
    Assume by contradiction that $E_i \neq K$, and let $j$ be the lowest index with $e_j\neq k_j$.
    Let $t_e$ be the time that $e_j$ was formed and $t_k$ be the time that $k_j$ was formed.
    At $\min(t_e,t_k)$, just before $e_j$ and $k_j$ were formed, $E_i$ and $K$ contained the same set of edges.
    Therefore the points that were not matched by $\alg$ until $\min(t_e,t_k)$, are the same in the two cases,
    and obviously the radii of the hemispheres at $\min(t_e,t_k)$ are the same in both cases as well.
    Thus, if $v_j$ and $u_j$ still exist in $\alg$'s run on $\mathcal{I}$ at that time, 
    and $v_j$ is on $u_j$'s hemisphere,
    then at the same time both $v_j$ and $u_j$ exist in $\alg$'s run on $C_i$, 
    and $v_j$ is on $u_j$'s hemisphere.
    Thus $\alg$ would match the pair $(u_j, v_j)$ when running on $C_i$ at $t_e = t_k$, 
    concluding $e_j = k_j$ and contradicting the assumption.
\end{proof}
\begin{corollary}\label{cor_cycles}
    By virtue of Lemma~\ref{lem_cycles} and Lemma~\ref{lem_local_opt_and_greedy} 
    it suffices to consider $\frac{\algoff}{\opt}$ when the union of the matchings produced by $\alg$ and $\opt$ forms a single cycle.
\end{corollary}

\begin{lemma}\label{lem_recurrence}
    Let $\gamma\in\doubleR$ s.t. $\gamma > 2$ and let $f:\doubleN \rightarrow \doubleR$ satisfy the recurrence relation
\begin{equation*}
    f(2k) = \min_{1\:\leq\:i\:\leq\:k - 1}\left\{f\left(2i\right),\:\frac{1}{\gamma}\left(f\left(2i\right) + f\left(2k-2i\right)\right)\right\},\;f(2) = 1
\end{equation*}
    
    Then, 
    \begin{equation*}
        f(n)=\Omega\left(
            \frac{1}{
                n^{\log\left(\frac{\gamma}{2}\right)}}
            \right)
    \end{equation*}
\end{lemma}
\begin{proof}
    We prove by induction on $k$ that $f(2k) \geq \left(\frac{2}{\gamma}\right)^{\log{k}}$.
    
    \

    \textbf{Base Case ($k = 1$):} $f(2)=1$, and $\left(\frac{2}{\gamma}\right)^{\log{1}} = \left(\frac{2}{\gamma}\right)^0 = 1$.
    
    \

    \textbf{Inductive step:} Assume the claim holds for all $j<k$.
    
    By the induction hypothesis for every $j < k$ it holds that
    $f(2j) \geq \left(\frac{2}{\gamma}\right)^{\log{j}} > \left(\frac{2}{\gamma}\right)^{\log{k}}$.
    Therefore, from the definition of $f$
    \begin{equation*}
        f(2k) \geq 
        \min\left(\left(\frac{2}{\gamma}\right)^{\log{k}},\:
        \frac{1}{\gamma}\left(f(2) + f(2k-2)\right),\:
        \frac{1}{\gamma}\left(f(4) + f(2k-4)\right),\:
        \ldots\right)
    \end{equation*}
    Define $h(j) = \frac{1}{\gamma}\left(f(2j) + f(2k-2j)\right)$, so 
    \begin{equation*}
        f(2k) \geq 
    \min\left(\left(\frac{2}{\gamma}\right)^{\log{k}},\:
    \min\limits_{1\:\leq\:j\:\leq\:k - 1}\left\{h(j)\right\}
    \right)
    \end{equation*}
    By the induction hypothesis,
    \begin{equation*}
        h(j) \geq \frac{1}{\gamma}\left(
        \left(\frac{2}{\gamma}\right)^{\log{j}} + \left(\frac{2}{\gamma}\right)^{\log{k-j}}
        \right)
        \geq \min\limits_{x\:\in\:\doubleR} \frac{1}{\gamma}\left\{
        \left(\frac{2}{\gamma}\right)^{\log{x}} + \left(\frac{2}{\gamma}\right)^{\log{k-x}}
        \right\}
    \end{equation*}
    
    $\left(\frac{2}{\gamma}\right)^{\log{x}} + \left(\frac{2}{\gamma}\right)^{\log{k-x}}$ is symmetric about $x = \frac{k}{2}$.
    Moreover, it is a concave function as it is the sum of two concave functions, thus the minimum point occurs at $x=\frac{k}{2}$.

    We found that $h(j) \geq \frac{1}{\gamma}\left(
        \left(\frac{2}{\gamma}\right)^{\log{\frac{k}{2}}} + \left(\frac{2}{\gamma}\right)^{\log{\frac{k}{2}}}
        \right) = \left(\frac{2}{\gamma}\right)^{\log{\frac{k}{2}}+1} = 
        \left(\frac{2}{\gamma}\right)^{\log{k}}
        $

    Hence, we conclude
    \begin{equation*}
        f(2k) \geq 
        \min\left(\left(\frac{2}{\gamma}\right)^{\log{k}},\:
        \left(\frac{2}{\gamma}\right)^{\log{k}}\right) = 
        \left(\frac{2}{\gamma}\right)^{\log{k}} = 
        \frac{1}{k^{\log{\frac{\gamma}{2}}}}    
    \end{equation*}
\end{proof}

\begin{lemma}\label{lem_algoff_to_opt}
    $\algoff \leq O\left(m^{\log\left(\frac{3+\epsilon}{2}\right)}\right) \opt$
\end{lemma}
\begin{proof}
    We view the requests as if they were in the time-augmented metric space $\mathcal{M}_T$, 
    and analyze the performance of $\alg$ in an offline manner.
    By Corollary~\ref{cor_cycles} we analyze the performance of $\alg$ 
    when $G=(\mathcal{I}, E)$, 
    the union of the matchings produced by $\alg$ and $\opt$, forms a single cycle.
    
    Denote $E_O$ the subset of edges matched by $\opt$, 
    and $E_A$ the subset of edges matched by $\alg$.
    Consider again the last two pairs of requests to be matched by $\alg$, 
    that is $(a,b)$ and $(c,d)$, and assume that $t_{ab}\leq t_{cd}$ and $t(b) \geq t(a)$
    ($t_{ab}$ is the time that $\alg$ matched $(a,b)$, and $t_{cd}$ is the time that $\alg$ matched $(c,d)$).
    Denote $T =\sum_{e\in E\setminus\{(c,d)\}}D(e)$, and let $O = \sum_{e\in E_O}D(e)$.
    From the triangle inequality we have that $D(c,d)$ is smaller than $T$, therefore
    \begin{equation}\label{eqn_T_O_relation}
        \frac{\algoff}{\opt} = \frac{D(c,d) + T - O}{O} \leq \frac{2T - O}{O} = 2\frac{T}{O} - 1
    \end{equation}
    We will bound $\frac{O}{T}$ from below, by developing and solving a recurrence relation similar to the one developed in \cite{ReingoldT81}, thus giving an upper bound on $\frac{\algoff}{\opt}$.

    Scale the distances so that $T=1$. Of course, $\frac{O}{T}$ stays the same.
    Let $f(m)$ be the minimal value of $\frac{O}{T}$
    over all possible inputs of size $m$ ($|\mathcal{I}| = m$), 
    when the union of the matchings produced by $\alg$ and $\opt$ forms a single cycle.
    
    For the sake of this analysis consider Figure~\ref{figure_analysis}.
    \begin{figure}[ht]
        \centering
        \begin{tikzpicture}
    [request/.style={circle,fill=black,thick, inner sep=0pt,minimum size=4pt}]
    \tikzstyle{every node}=[font=\small]

    \node (a) [request]  at (0,0) [label=below:$a$]{};
    \node (b) [request] at (2,0) [label=below:$b$]{};
    \node (c) [request] at (-0.5,4) [label=above:$c$]{};
    \node (d) [request] at (2.5,4) [label=above:$d$]{};
    \draw[black, dash dot] (a) to[out=115,in=-115] node[midway,fill=white] {$P_{ca}$} (c);
    \draw[black, dash dot] (b) to[out=65,in=-65] node[midway,fill=white] {$P_{db}$} (d);
    \draw[black] (a) -- (b);
    \draw[black] (c) -- (d);
\end{tikzpicture}
        \caption[0]{The cycle formed by the union of the matchings produced by $\alg$ and $\opt$. 
        \hspace{\textwidth} The length of $P_{ca}$ is $\alpha$, and the length of $P_{db}$ is $\beta$.}
        \label{figure_analysis}
    \end{figure}
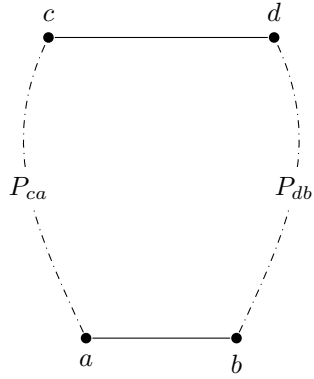

    Let $P_{ca}$ be the alternating path from $c$ to $a$, and $P_{db}$ 
    be the alternating path from $d$ to $b$.
    Denote $\alpha = \sum_{e\in P_{ca}}D(e)$, and $\beta = \sum_{e\in P_{db}}D(e)$.
    Then, by the triangle inequality 
    \begin{equation}\label{eqn_triangle}
        \alpha \geq D(a,c)
    \end{equation}
    From Lemma~\ref{lem_main} we have
    \begin{equation}\label{eqn_lem_main}
        (1+\epsilon)D(a,c) \geq D(a,b)
    \end{equation}
    It follows from Equations (\ref{eqn_triangle}) and (\ref{eqn_lem_main}) that
    \begin{equation}
        1-\alpha-\beta=D(a,b)\leq (1+\epsilon)\alpha
    \end{equation} 
    Similarly $1-\alpha-\beta \leq (1+\epsilon)\beta$.

    Let $2i$ be the number of points on $P_{ca}$, then
    $f(m)$ satisfies the recurrence relation
    \begin{equation}
        f(m) = \min_{\substack{
            1 \:\leq\: i \:<\: \frac{m}{2} - 1\\
            0 \:<\: 1 - \alpha - \beta \:\leq\: (1 + \epsilon)\alpha\\
            0 \:<\: 1 - \alpha - \beta \:\leq\: (1 + \epsilon)\beta}} 
            \{\alpha f(2i) + \beta f(m-2i)\}
    \end{equation}
    
    Conditioning on $t$, $f(t)$ and $f(m-t)$ are constant, 
    therefore $\alpha f(t) + \beta f(m-t)$ becomes a linear function in $\alpha$ and $\beta$, 
    so its minimum must occur at a vertex of the polyhedron defined by the minimization constraints
    (see for example \cite{Dantzig63}).
    
    The vertices of this polyhedron are $(1,0), (0,1), (\frac{1}{3 + \epsilon}, \frac{1}{3 + \epsilon})$, so
    \begin{equation}
        f(m) = \min_{1 \:\leq\: i \:\leq\: \frac{m}{2} - 1}\left\{f\left(2i\right),\:\frac{1}{3 + \epsilon}\left(f\left(2i\right) + f\left(m-2i\right)\right)\right\}
    \end{equation}
    Also note that $f(2) = 1$, since there is only one way to match two points, so $T = O$.
    The conditions of Lemma~\ref{lem_recurrence} are met with $\gamma = 3 + \epsilon$, 
    thus 
    \begin{equation*}
        f(m)=\Omega\left(\frac{1}{m^{\log\left(\frac{3+\epsilon}{2}\right)}}\right)
    \end{equation*}
    Finally, from \ref{eqn_T_O_relation} we conclude
    \begin{equation*}
        \frac{\algoff}{\opt} \leq 2\frac{T}{O} - 1 \leq \frac{2}{f(m)} = O\left(m^{\log\left(\frac{3+\epsilon}{2}\right)}\right)
    \end{equation*}
\end{proof}    
\begin{lemma}\label{lem_epsilon}
    $\algon = \Theta\left(\frac{1}{\epsilon}\right) \algoff$
\end{lemma}
\begin{proof}
    Assume two requests $p$ and $q$ were matched by $\alg$ at time $t$.
    Assume w.l.o.g that $t(p) \geq t(q)$.
    The contribution of this pair to $\algon$, is
    \begin{equation*}
        \begin{split}
            t - t(p) + t - t(q) + d(x(p),x(q)) &= \\
            t - t(p) + t - t(p) + t(p) - t(q) + d(x(p),x(q)) &= 
            2(t - t(p)) + D(p,q)    
        \end{split}
    \end{equation*}
    On the contrary, the contribution of this pair to $\algoff$, is just $D(p,q)$.

    Note that $t$ is the time that $q$ was on $p$'s hemisphere, 
    so $t = t(p) + \frac{D(p,q)}{\epsilon}$,
    hence the ratio between $\algon$ and $\algoff$ for this pair is
    \begin{equation*}
        \frac{2\frac{D(p,q)}{\epsilon} + D(p,q)}{D(p,q)} = 1 + \frac{2}{\epsilon}
    \end{equation*}
    Summing over all matched pairs we get 
    $\frac{\algon}{\algoff} = 1 + \frac{2}{\epsilon} = \Theta\left(\frac{1}{\epsilon}\right)$.
\end{proof}

Finally we prove Theorem \ref{thm_main} using the inequalities proven in the previous lemmas.
\begin{proof}[Proof of Theorem \ref{thm_main}]
    Combining Lemma~\ref{lem_opts}, Lemma~\ref{lem_algoff_to_opt} and Lemma~\ref{lem_epsilon} we have
    \begin{equation*}
        \algon \leq O\left(\frac{1}{\epsilon}\right)\algoff \leq O\left(\frac{1}{\epsilon}m^{\log\left(\frac{3+\epsilon}{2}\right)}\right) \opt        
    \end{equation*}
    Hence, $\alg(\epsilon)$ is $O\left(\frac{1}{\epsilon}m^{\log\left(\frac{3+\epsilon}{2}\right)}\right)$-competitive.
\end{proof}

In Appendix~\ref{sec_lower_bound} we show that the analysis is tight, 
and that the competitive ratio is indeed a function of $m$, 
and not of $n$ (the size of the metric space).
In Appendix~\ref{sec_time} we show that growing hemispheres in space while ignoring the time axis,
and other similar hacks, only worsen the competitive ratio.

\section{The Bipartite Case}\label{sec_bipartite}
For the bipartite case, we suggest the same algorithm as in the monochromatic case.
The only difference is that we match a request 
$q$ to a request $p$ as soon as $q$ is found on the boundary of $p$'s hemisphere, 
\textbf{and that $q$ and $p$ do not belong to the same class}.

\begin{algorithm}
    \caption{A Deterministic Algorithm for MBPMD on General Metrics\label{alg:bipartite}}
    \begin{algorithmic}[1]
        \Procedure{ALG-B}{$\epsilon$}
        \State \textbf{At every moment $t$:}
        \State Add the new requests that arrive at time $t$
        \ForEach{unmatched request $p$}
            \ForEach{unmatched request $q \neq p$}
                \If{$t(p) \geq t(q)$ \textbf{and} $t = t(p) + \frac{D\left(x(p),x(q)\right)}{\epsilon}$ \textbf{and} $class(q) \neq class(p)$}
                \State \textbf{match}$(p,q)$
                \EndIf
            \EndFor
        \EndFor
     \EndProcedure
    \end{algorithmic}
\end{algorithm}
\subsection{Analysis}
We prove the following theorem:
\begin{theorem}\label{thm_main_bipartite}
    $\algb(\epsilon)$ is $O\left(\frac{1}{\epsilon}m^{\log\left(\frac{3+\epsilon}{2}\right)}\right)$-competitive.
\end{theorem}

Observation~\ref{obs1}, Lemma~\ref{lem_cycles} and Lemma~\ref{lem_local_opt_and_greedy} hold for the bipartite case as well, 
therefore using Corollary~\ref{cor_cycles} we may assume that the union of $\algb$ and $\opt$ forms a single cycle.

The key difference in the analysis for this case,
is that when we consider the last four requests to be matched,
not every two of them could have been matched to each other.
Therefore Lemma~\ref{lem_main} does not hold, but a weaker yet similar result does.

Consider the last two pairs of requests to be matched by $\algb$.
Name them $(a,b)$ and $(c,d)$, 
and assume w.l.o.g that $(a,b)$ were matched at time $t_{ab}$,
and $(c,d)$ at $t_{cd} \geq t_{ab}$.
Also, assume w.l.o.g that $t(a) \leq t(b)$.
\begin{lemma}\label{lem_main_bipartite}
    If $class(a) = class(d) \neq class(b) = class(c)$ then
    \begin{enumerate}
        \item $D(a,b) \leq (1+\epsilon)D(a,c)$\label{lem_main_bipartite:1}
        \item $D(a,b) \leq (1+\epsilon)D(b,d)$\label{lem_main_bipartite:2}
    \end{enumerate}
\end{lemma}
We omit the proof of this lemma as it is the same as the proof of Lemma~\ref{lem_main} for the relevant cases.

Considering Figure~\ref{figure_analysis} we have the following lemma.
\begin{lemma}\label{lem_classes}
    $class(a) = class(d) \neq class(b) = class(c)$
\end{lemma}
\begin{proof}
    From the alternation property of Observation~\ref{obs1} we have that the number of edges along $P_{ca}$ must be odd 
    (since the number of $\opt$ edges along $P_{ca}$ must be one more than $\algb$ edges along $P_{ca}$).
    Moreover, the classes of the requests along $P_{ca}$ alternate as well 
    (since every edge must match requests of different classes).
    Since there are odd number of edges along $P_{ca}$, 
    there are odd number of class alternations along $P_{ca}$, 
    so the class of the last request along $P_{ca}$ (that is $class(c)$)
    must be different from the class of the first request along $P_{ca}$ (that is $class(a)$).
    Thus $class(c) \neq class(a)$ and of course $class(a) \neq class(b)$, $class(c) \neq class(d)$, 
    so $class(a) = class(d) \neq class(b) = class(c)$.
\end{proof}
Using Lemma~\ref{lem_classes} and Lemma~\ref{lem_main_bipartite} 
we repeat the proof of Lemma~\ref{lem_algoff_to_opt} and achieve the following result:
\begin{lemma}\label{lem_algboff_to_opt}
    $\algboff \leq O\left(m^{\log\left(\frac{3+\epsilon}{2}\right)}\right) \opt$
\end{lemma}
The main theorem for the bipartite case now follows:
\begin{proof}[Proof of Theorem~\ref{thm_main_bipartite}]
    Lemma~\ref{lem_epsilon} and Lemma~\ref{lem_opts} hold for $\algb$ as well,
    thus from Lemma~\ref{lem_algboff_to_opt} we have
    \begin{equation*}
        \algbon \leq O\left(\frac{1}{\epsilon}\right)\algboff \leq O\left(\frac{1}{\epsilon}m^{\log\left(\frac{3+\epsilon}{2}\right)}\right) \opt
    \end{equation*}
    Hence, $\algb(\epsilon)$ is $O\left(\frac{1}{\epsilon}m^{\log\left(\frac{3+\epsilon}{2}\right)}\right)$-competitive.
\end{proof}

\section{Concluding Remarks and Open Problems}\label{sec_rem}

In this paper we presented the first sub-linear competitive deterministic
algorithm for \textit{Minimum-Cost Perfect Matching with Delays} as a function of $m$, 
the number of requests. 
We also provided a similar algorithm for the problem
of \textit{Minimum-Cost Bipartite Perfect Matching with Delays}
achieving the same competitive ratio.

One open problem is to decide if a deterministic algorithm with a better competitive ratio exists,
in particular a $\polylog(m)$-competitive one, 
by showing a lower bound or providing an algorithm for the problem.
In addition, 
the problem of finding a sub-linear in $n$ 
competitive deterministic algorithm is still open.

\bibliographystyle{plain}
\bibliography{references}

\begin{thebibliography}{10}

\bibitem{AntoniadisBNPS14}
Antonios Antoniadis, Neal Barcelo, Michael Nugent, Kirk Pruhs, and Michele
  Scquizzato.
\newblock A $o(n)$-competitive deterministic algorithm for online matching on a
  line.
\newblock In {\em Approximation and Online Algorithms - 12th International
  Workshop}, pages 11--22, 2014.

\bibitem{Ashlagi17}
Itai Ashlagi, Yossi Azar, Moses Charikar, Ashish Chiplunkar, Ofir Geri, Haim
  Kaplan, Rahul~M. Makhijani, Yuyi Wang, and Roger Wattenhofer.
\newblock Min-cost bipartite perfect matching with delays.
\newblock In {\em Approximation, Randomization, and Combinatorial Optimization.
  Algorithms and Techniques}, pages 1:1--1:20, 2017.

\bibitem{Azar17}
Yossi Azar, Ashish Chiplunkar, and Haim Kaplan.
\newblock Polylogarithmic bounds on the competitiveness of min-cost perfect
  matching with delays.
\newblock In {\em Proceedings of the Twenty-Eighth Annual {ACM-SIAM} Symposium
  on Discrete Algorithms}, pages 1051--1061, 2017.

\bibitem{Azar17Service}
Yossi Azar, Arun Ganesh, Rong Ge, and Debmalya Panigrahi.
\newblock Online service with delay.
\newblock In {\em Proceedings of the 49th Annual {ACM} {SIGACT} Symposium on
  Theory of Computing}, pages 551--563, 2017.

\bibitem{BansalBGN_Algorithmica14}
Nikhil Bansal, Niv Buchbinder, Anupam Gupta, and Joseph Naor.
\newblock A randomized ${O}(log^2k)$-competitive algorithm for metric bipartite
  matching.
\newblock {\em Algorithmica}, 68(2):390--403, 2014.

\bibitem{Bienkowski18}
Marcin Bienkowski, Artur Kraska, Hsiang-Hsuan Liu, and Pawel Schmidt.
\newblock A primal-dual online deterministic algorithm for matching with
  delays.
\newblock {\em CoRR}, abs/1804.08097, 2018.

\bibitem{BienkowskiKS17}
Marcin Bienkowski, Artur Kraska, and Pawel Schmidt.
\newblock A match in time saves nine: Deterministic online matching with
  delays.
\newblock In {\em Approximation and Online Algorithms - 15th International
  Workshop}, pages 132--146, 2017.

\bibitem{Dantzig63}
George~B. Dantzig.
\newblock {\em Linear programming and extensions}.
\newblock Princeton University Press, 1963.

\bibitem{Edmonds_CJM65}
Jack Edmonds.
\newblock Paths, trees, and flowers.
\newblock {\em Canadian Journal of Mathematics}, 17:449--467, 1965.

\bibitem{EmekKW_STOC16}
Yuval Emek, Shay Kutten, and Roger Wattenhofer.
\newblock Online matching: haste makes waste!
\newblock In {\em Proceedings of the 48th Annual {ACM} {SIGACT} Symposium on
  Theory of Computing}, pages 333--344, 2016.

\bibitem{EmekSW17}
Yuval Emek, Yaacov Shapiro, and Yuyi Wang.
\newblock Minimum cost perfect matching with delays for two sources.
\newblock In {\em Algorithms and Complexity - 10th International Conference},
  pages 209--221, 2017.

\bibitem{FuchsHK03}
Bernhard Fuchs, Winfried Hochst{\"{a}}ttler, and Walter Kern.
\newblock Online matching on a line.
\newblock {\em Electronic Notes in Discrete Mathematics}, 13:49--51, 2003.

\bibitem{GuptaL12}
Anupam Gupta and Kevin Lewi.
\newblock The online metric matching problem for doubling metrics.
\newblock In {\em Automata, Languages, and Programming - 39th International
  Colloquium}, pages 424--435, 2012.

\bibitem{KalyanasundaramP_JAlgorithms93}
Bala Kalyanasundaram and Kirk Pruhs.
\newblock Online weighted matching.
\newblock {\em J. Algorithms}, 14(3):478--488, 1993.

\bibitem{kalyanasundaram_pruhs_1998}
Bala Kalyanasundaram and Kirk Pruhs.
\newblock On-line network optimization problems.
\newblock In {\em Online Algorithms, The State of the Art (the book grow out of
  a Dagstuhl Seminar)}, pages 268--280, 1996.

\bibitem{KhullerMV91}
Samir Khuller, Stephen~G. Mitchell, and Vijay~V. Vazirani.
\newblock On-line algorithms for weighted bipartite matching and stable
  marriages.
\newblock In {\em Automata, Languages and Programming, 18th International
  Colloquium}, pages 728--738, 1991.

\bibitem{KoutsoupiasN03}
Elias Koutsoupias and Akash Nanavati.
\newblock The online matching problem on a line.
\newblock In {\em Approximation and Online Algorithms, First International
  Workshop}, pages 179--191, 2003.

\bibitem{MeyersonNP06}
Adam Meyerson, Akash Nanavati, and Laura~J. Poplawski.
\newblock Randomized online algorithms for minimum metric bipartite matching.
\newblock In {\em Proceedings of the Seventeenth Annual {ACM-SIAM} Symposium on
  Discrete Algorithms}, pages 954--959, 2006.

\bibitem{NayyarR17}
Krati Nayyar and Sharath Raghvendra.
\newblock An input sensitive online algorithm for the metric bipartite matching
  problem.
\newblock In {\em 58th {IEEE} Annual Symposium on Foundations of Computer
  Science}, pages 505--515, 2017.

\bibitem{raghvendra16}
Sharath Raghvendra.
\newblock A robust and optimal online algorithm for minimum metric bipartite
  matching.
\newblock In {\em Approximation, Randomization, and Combinatorial Optimization.
  Algorithms and Techniques}, pages 18:1--18:16, 2016.

\bibitem{ReingoldT81}
Edward~M. Reingold and Robert~Endre Tarjan.
\newblock On a greedy heuristic for complete matching.
\newblock {\em {SIAM} J. Comput.}, 10(4):676--681, 1981.

\end{thebibliography}
\appendix
\section{The competitive ratio is a function of \texorpdfstring{$m$}{m}}\label{sec_lower_bound}
Following Section~\ref{sec_analysis}, a question arises - whether Theorem~\ref{thm_main} can be modified to prove $\frac{\alg}{\opt} \leq O\left(\frac{1}{\epsilon}n^{\log\left(\frac{3}{2}+\epsilon\right)}\right)$ for a finite metric space of size $n$.

We show that for every $\alg\left(\epsilon\right)$ there is an instance with $n = 1$ 
for which $\frac{\alg}{\opt} \geq \lowerbound$.
The instance we give is essentially the example given by~\cite{ReingoldT81}, 
over the time axis, and with distances scaled to consider the progress of time.
Let $k = \log(m)$ and consider Figure~\ref{figure_lower_bound} which describes a series of requests with the recurrence relation
\begin{equation}\label{eqn_rec_relation}
    a_i = \frac{b_i}{1 + \epsilon}, \; b_i = 2b_{i-1} + a_{i-1}, \;  b_1 = 1
\end{equation}
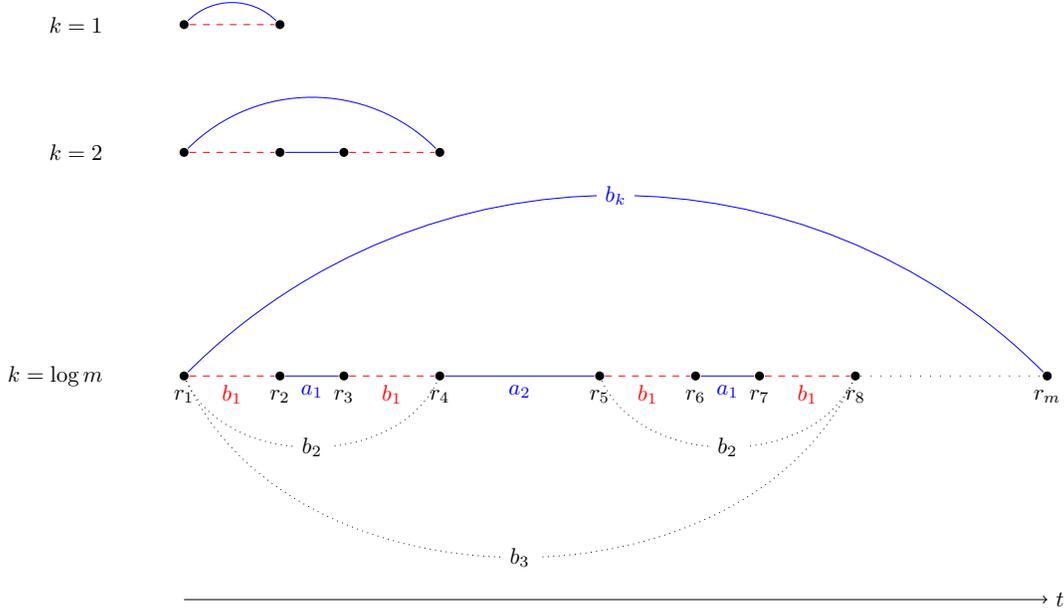
\begin{figure}[ht]
    \centering
    \scalebox{0.85}
    {
    \begin{tikzpicture}
    [request/.style={circle,fill=black,thick, inner sep=0pt,minimum size=4pt}]
    \tikzstyle{every node}=[font=\small]
    % time axis
    \draw[->] (0,-8) -- (13.5,-8) node[right] {$t$}; 

    % k = 1
    \node (r1k1) [request]  at (0,1) {};
    \node (r2k1) [request] at (1.5,1) {};
    \draw[red, dashed] (r1k1) -- (r2k1);
    \draw[blue] (r1k1) to[out=45,in=135] (r2k1);

    \node (k1label) [label=left:{$k = 1$}, left of=r1k1] {};
    % k = 2
    \node (r1k2) [request]  at (0,-1) {};
    \node (r2k2) [request] at (1.5,-1) {};
    \node (r3k2) [request] at (2.5,-1) {};
    \node (r4k2) [request] at (4,-1) {};
    \draw[red, dashed] (r1k2) -- (r2k2);
    \draw[blue] (r2k2) -- (r3k2);
    \draw[red, dashed] (r3k2) -- (r4k2);
    \draw[blue] (r1k2) to[out=45,in=135] (r4k2);
    
    \node (k2label) [label=left:{$k = 2$}, left of=r1k2] {};
    % k = logm
    \node (r1k) [request]  at (0,-4.5) [label=below:$r_1$]{};
    \node (r2k) [request] at (1.5,-4.5) [label=below:$r_2$]{};
    \node (r3k) [request] at (2.5,-4.5) [label=below:$r_3$]{};
    \node (r4k) [request] at (4,-4.5) [label=below:$r_4$]{};
    \node (r5k) [request]  at (6.5,-4.5) [label=below:$r_5$]{};
    \node (r6k) [request] at (8,-4.5) [label=below:$r_6$]{};
    \node (r7k) [request] at (9,-4.5) [label=below:$r_7$]{};
    \node (r8k) [request] at (10.5,-4.5) [label=below:$r_8$]{};
    \node (rmk) [request] at (13.5,-4.5) [label=below:$r_m$]{};
    
    \draw[red, dashed] (r1k) -- (r2k) node[below, midway] {$b_1$};
    \draw[blue] (r2k) -- (r3k) node[below, midway] {$a_1$};
    \draw[red, dashed] (r3k) -- (r4k) node[below, midway] {$b_1$};
    \draw[blue] (r4k) -- (r5k) node[below, midway] {$a_2$};
    \draw[red, dashed] (r5k) -- (r6k) node[below, midway] {$b_1$};
    \draw[blue] (r6k) -- (r7k) node[below, midway] {$a_1$};
    \draw[red, dashed] (r7k) -- (r8k) node[below, midway] {$b_1$};
    \draw[black, loosely dotted] (r8k) -- (rmk);
    \draw[blue] (r1k) to[out=45,in=135] node[midway,fill=white] {$b_k$} (rmk);

    \draw[black, dotted] (r1k) to[out=-65,in=-115] node[midway,fill=white] {$b_2$} (r4k);
    \draw[black, dotted] (r5k) to[out=-65,in=-115] node[midway,fill=white] {$b_2$} (r8k);
    \draw[black, dotted] (r1k) to[out=-65,in=-115] node[midway,fill=white] {$b_3$} (r8k);

    \node (kmlabel) [label=left:{$k = \log{m}$}, left of=r1k] {};
\end{tikzpicture}
    }
    \caption[1]{A series of $m$ requests along the time axis with $n = 1$ and $\frac{\alg}{\opt} \geq \lowerbound$. \hspace{\textwidth} In blue is the matching produced by $\alg$, and in dashed red - a matching of cost $O(m)$.}
    \label{figure_lower_bound}
\end{figure}

\begin{lemma}\label{lem_lower_alg}
    $\alg$ matches $\left(r_2,r_3\right)\ldots\left(r_{m-2},r_{m-1}\right)$ and $\left(r_1,r_m\right)$.
\end{lemma}
\begin{proof}
    We prove the lemma by induction on $k$.
    
    \

    \textbf{Base Case ($k = 1$):} The only point that $r_2 = r_m$ can be matched to is $r_1$.
    
    \

    \textbf{Inductive step:} 
    Assume the claim holds for $k - 1$.
    We start by showing that $\alg$ will match the pairs $(r_2, r_3), \ldots, (r_{\frac{m}{2} - 2}, r_{\frac{m}{2} - 1})$.

    Let $t_0 = b_{k-1} + \frac{b_{k-1}}{\epsilon}$, this is the time that the hemisphere of $r_{\frac{m}{2}}$ reaches $r_1$ 
    unless $r_{\frac{m}{2}}$ is matched by another request at time $t < t_0$.
    By the induction hypothesis, 
    unless the hemisphere of some $r_i$ with $i > \frac{m}{2}$ reaches past $r_{\frac{m}{2}}$ by $t < t_0$,
    the hemisphere of $r_{\frac{m}{2}}$ will reach $r_1$, after the pairs $(r_2, r_3)$,\ldots,$(r_{\frac{m}{2} - 2}, r_{\frac{m}{2} - 1})$ are matched.
    Notice that the hemisphere of $r_{\frac{m}{2}}$ may reach $r_1$ only by time $t_0$
    and the hemisphere of $r_{\frac{m}{2} + 1}$ may reach $r_{\frac{m}{2}}$ only by 
    \begin{equation*}
        t_1 = b_{k-1} + a_{k-1} + \frac{a_{k-1}}{\epsilon} = b_{k-1} + \frac{b_{k-1}}{1+\epsilon}(1 + \frac{1}{\epsilon}) = t_0
    \end{equation*}
    Therefore, the hemisphere of $r_{\frac{m}{2} + 1}$ may reach $r_{\frac{m}{2}}$ 
    only after $(r_2, r_3)$,\ldots,$(r_{\frac{m}{2} - 2}, r_{\frac{m}{2} - 1})$ are matched.
    Obviously for every $i > \frac{m}{2} + 1$ the hemisphere of $r_i$ would not reach past $r_{\frac{m}{2}}$ by $t_0$ if the hemisphere of $r_{\frac{m}{2} + 1}$ does not, 
    therefore $(r_2, r_3)$,\ldots,$(r_{\frac{m}{2} - 2}, r_{\frac{m}{2} - 1})$ are matched by $\alg$ by time $t_0$.
    
    Considering $r_{\frac{m}{2} + 1}$,\ldots,$r_{m}$, again by the induction hypothesis we have that
    unless $r_{m}$ is matched by another request before its hemisphere reaches $r_{\frac{m}{2} + 1}$, 
    $\alg$ will match the pairs $(r_{\frac{m}{2} + 2}, r_{\frac{m}{2} + 3})$, \ldots, $(r_{m-2}, r_{m-1})$. 
    Indeed, there is no request after $r_m$, thus $\alg$ will match these pairs,
    and we are left to address the requests $r_1$, $r_{\frac{m}{2}}$, $r_{\frac{m}{2} + 1}$, $r_m$.
    
    Observe that the hemisphere of $r_m$ reaches $r_{\frac{m}{2} + 1}$ 
    at $t = b_k + \frac{b_{k-1}}{\epsilon} > (1 + \frac{1}{\epsilon})b_{k-1} = t_1 = t_0$,
    hence $\alg$ will match the pair $(r_{\frac{m}{2} + 1}, r_{\frac{m}{2}})$.
    The remaining and last pair to be matched by $\alg$ is $(r_1, r_m)$ of course.
\end{proof}

The cost of $\opt$ is at most $O(m)$ since $D(u,v) = b_1 = 1$ 
for every pair $(u,v)$ in the matching $\left(r_1,r_2\right)\ldots\left(r_{m-1},r_m\right)$.
The cost of matching $r_1$ to $r_m$ is $b_k$.
Out of the pairs $\left(r_2,r_3\right)\ldots\left(r_{m-2},r_{m-1}\right)$ 
there are $2^i$ pairs with distance $a_{k-i-1}$ between the two end-points, 
for $0 \leq i \leq k - 2$.
Therefore $\algoff = b_k + \sum_{i=0}^{k-2}2^{i}a_{k-i-1}$.

The mutual recurrence relation (\ref{eqn_rec_relation}) solves to
\begin{equation}
    a_i = \frac{\left(2 + \frac{1}{1 + \epsilon}\right)^i}{2\epsilon + 3}, \; b_i = \left(2 + \frac{1}{1 + \epsilon}\right)^{i-1}
\end{equation}
Therefore
\begin{align*}
    \algoff = b_k + \sum_{i=1}^{k-1}2^{k-1-i}a_{i} &> \sum_{i=1}^{k-1}\frac{\left(2 + \frac{1}{1 + \epsilon}\right)^i}{2\epsilon + 3}2^{k-1-i} \\
    &=\frac{2^{k-1}}{2\epsilon + 3}\sum_{i=1}^{k-1}\left(1 + \frac{1}{2\left(1 + \epsilon\right)}\right)^i \\
    &=\frac{2^{k-1}}{2\epsilon + 3}\left(2\epsilon + 3\right)\left(\left(1 + \frac{1}{2\left(1 + \epsilon\right)}\right)^{k-1} - 1 \right) \\
    &=\left(2+\frac{1}{1+\epsilon}\right)^{k-1} - 2^{k-1}
\end{align*}
Hence, 
\begin{equation*}
    \frac{\algoff}{\opt} \geq \frac{\left(2+\frac{1}{1+\epsilon}\right)^{k-1} - 2^{k-1}}{2^k} =
    \Omega\left(m^{\log\left(1+\frac{1}{2(1+\epsilon)}\right)}\right) = \Omega\left(m^{\log\left(\frac{3+2\epsilon}{2+2\epsilon}\right)}\right)
\end{equation*}
Finally, from Lemma~\ref{lem_epsilon} we have 
\begin{equation*}
    \frac{\algon}{\opt} = \Omega\left(\frac{1}{\epsilon}\right)\frac{\algoff}{\opt} \geq \lowerbound
\end{equation*}

\section{Time must be considered}\label{sec_time}
A simple hack that may handle the instance given in Appendix~\ref{sec_lower_bound},
is to match immediately two points that are located at the same position in space.
Obviously this will not handle some very similar instances,
generated by small perturbations of the positions of the requests.

A simple extension of this idea is to ignore the time axis, 
so that $p$ and $q$ will be matched as soon as $t \geq \min\left(t(p), t(q)\right) + \frac{d(x(p),x(q))}{\epsilon}$,
i.e. the requests grow spheres only in space, but not in time, 
and they are matched to each other as soon as one of them is in the sphere of the other.

The instance in Figure \ref{figure_lower_bound_2} shows that the competitive-ratio of this algorithm can be worse as $\Omega(m)$,
even though the size of the metric space is $n = 2$.
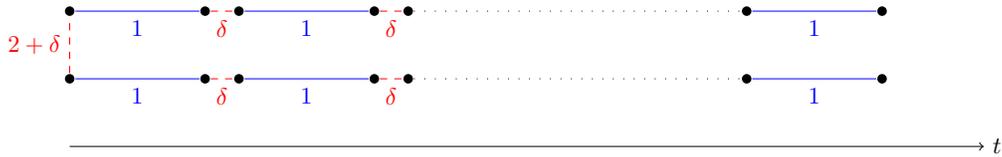
\begin{figure}[H]
    \centering
    \scalebox{0.9}
    {
    \begin{tikzpicture}
    [request/.style={circle,fill=black,thick, inner sep=0pt,minimum size=4pt}]
    \tikzstyle{every node}=[font=\small]
    % time axis
    \draw[->] (0,-2) -- (13.5,-2) node[right] {$t$}; 

    % k = logm
    \node (r1ka) [request]  at (0,-1) {};
    \node (r2ka) [request] at (2,-1) {};
    \node (r3ka) [request] at (2.5,-1) {};
    \node (r4ka) [request] at (4.5,-1) {};
    \node (r5ka) [request]  at (5,-1) {};
    \node (rm_ka) [request] at (10,-1) {};
    \node (rmka) [request] at (12,-1) {};

    \node (r1kb) [request]  at (0,0) {};
    \node (r2kb) [request] at (2,0) {};
    \node (r3kb) [request] at (2.5,0) {};
    \node (r4kb) [request] at (4.5,0) {};
    \node (r5kb) [request]  at (5,0) {};
    \node (rm_kb) [request] at (10,0) {};
    \node (rmkb) [request] at (12,0) {};
    
    \draw[blue] (r1ka) -- (r2ka) node[below, midway] {$1$};
    \draw[red, dashed] (r2ka) -- (r3ka) node[below, midway] {$\delta$};
    \draw[blue] (r3ka) -- (r4ka) node[below, midway] {$1$};
    \draw[red, dashed] (r4ka) -- (r5ka) node[below, midway] {$\delta$};
    \draw[black, loosely dotted] (r5ka) -- (rm_ka);
    \draw[blue] (rm_ka) -- (rmka) node[below, midway] {$1$};

    \draw[blue] (r1kb) -- (r2kb) node[below, midway] {$1$};
    \draw[red, dashed] (r2kb) -- (r3kb) node[below, midway] {$\delta$};
    \draw[blue] (r3kb) -- (r4kb) node[below, midway] {$1$};
    \draw[red, dashed] (r4kb) -- (r5kb) node[below, midway] {$\delta$};
    \draw[black, loosely dotted] (r5kb) -- (rm_kb);
    \draw[blue] (rm_kb) -- (rmkb) node[below, midway] {$1$};

    \draw[red, dashed] (r1ka) -- (r1kb) node[left, midway] {$2 + \delta$};
\end{tikzpicture}
    }
    \caption[2]{In blue - the matching produced by the suggested algorithm, of cost $O(m)$.\hspace{\textwidth}In dashed red - an alternative matching of cost $O({1 + \delta}m)$.}
    \label{figure_lower_bound_2}
\end{figure}

Note that $\Omega(m)$ competitive-ratio will be achieved for this instance,
even for similar algorithms which do not consider time, 
such as matching $p$ to $q$ if $t(p) \geq t(q)$ and $t \geq t(p) + \frac{d(x(p),x(q))}{\epsilon}$,
or matching $p$ to $q$ if $t(p) \leq t(q)$ and $t \geq t(p) + \frac{d(x(p),x(q))}{\epsilon}$.

\end{document}